\documentclass[reqno,a4paper,12pt]{amsart}

\usepackage{amsmath,amstext,amsfonts,amsbsy,eucal,amssymb}
\usepackage[latin1]{inputenc}
\usepackage{float} 

\numberwithin{equation}{section}

\newtheorem{theorem}{Theorem}[section]
\newtheorem{lemma}[theorem]{Lemma}

\newtheorem{proposition}[theorem]{Proposition}

\newtheorem{conjecture}[theorem]{Conjecture}

\theoremstyle{definition}
\newtheorem{definition}[theorem]{Definition}
\newtheorem{example}[theorem]{Example}
\newtheorem{remark}[theorem]{Remark}

\begin{document}

\parskip 4pt
\baselineskip 16pt


\title[An infinite family of homogeneous discrete equations]
{An infinite family of homogeneous discrete equations with the Laurent property}

\author[Andrei K. Svinin]{Andrei K. Svinin}
\address{Andrei K. Svinin, 
Matrosov Institute for System Dynamics and Control Theory of 
Siberian Branch of Russian Academy of Sciences,
P.O. Box 292, 664033 Irkutsk, Russia}
\email{svinin@icc.ru}

%
%

\date{\today}



\begin{abstract}
We present and investigate a new infinite family of homogeneous  equations which possess the Laurent property. The first representative in this family is the well-known  Somos-5 recurrence. 
\end{abstract}

\maketitle

\section*{Introduction}

Recently, considerable interest has been shown in nonlinear equations of the form
\begin{equation}
t_nt_{n+N}=L\left(t_{n+1},\ldots, t_{n+N-1}\right),
\label{887777}
\end{equation}
which possess the Laurent property \cite{Alman}, \cite{Fomin}, \cite{Fordy}, \cite{Lam}, where $L$ is assumed to be a (Laurent) polynomial.
\begin{definition} \label{89765}
Equation (\ref{887777}) is said to have the Laurent property if $t_n\in\mathbb{Z}[t_0^{\pm 1},\ldots, t_{N-1}^{\pm 1}], \forall n\in\mathbb{Z}$.
\end{definition}
The requirement for nonlinear  equation formulated by this definition  seems extremely strict, and one is entitled to ask whether such equations even exist, but as it turns out, they do not just exist -- there are actually quite a lot of such recurrences. A classic example of such equations represents the class of Gale-Robinson recurrences which are written in the form
\begin{equation}
t_{n}t_{n+N}=\alpha t_{n+a}t_{n+N-a}+\beta t_{n+c}t_{n+N-c},
\label{887}
\end{equation}
where $N\geq 4$ and $1\leq a<c\leq \left\lfloor N/2 \right\rfloor$. Note that $\alpha$ and $\beta$, in (\ref{887}) are assumed to be arbitrary parameters, but in applications, as a rule, they are equal to one.

It is clear that there is something behind the Laurent property. It is currently known that this property is linked to two key developments: the theory of cluster algebras by Fomin and Zelevinsky \cite{Fomin} and the broader framework of Laurent phenomenon algebras introduced by Lam and Pylyavskyy \cite{Lam}.
In particular, Gale-Robinson recurrences provide a concrete example of a dynamical system that is fully described in the language of cluster algebras: it arises from quiver mutations, its Laurent property follows from general theorems of cluster algebras. Fomin and Zelevinsky, in their seminal work \cite{Fomin}, used precisely the Gale-Robinson recurrence (\ref{887}) as one of the main examples to demonstrate and prove this property within the framework of their  theory.  Thus, cluster algebras provide a universal explanation for the \textit{mysterious} Laurent property for this and similar classes of recurrences.

If the Laurent property is already known for some nonlinear recurrence, then an elementary consequence is the integrality of the terms of the sequence defined by this recurrence under certain conditions. Indeed, if we put $t_j=\pm 1$ for $j=0,\ldots, N-1$, then we can be sure of the integerness of the corresponding sequence. In some cases, nonlinear recurrences generate integer sequences with a sufficiently large number of applications. Nonlinearity here is not an obstacle, but rather a source of rich structure; integrality is not a coincidence, but a manifestation of deep algebraic and geometric laws.

Allow us to delve a little deeper into the history. In the early  1980s, Michael Somos discovered several nonlinear recurrence relations which, despite their nonlinearity, generate integer sequences for suitable initial data. One of them is the equation
\begin{equation}
t_nt_{n+5}=\alpha t_{n+1}t_{n+4}+\beta t_{n+2}t_{n+3}
\label{889987}
\end{equation}
known in the literature as Somos-5. It is clear that (\ref{889987}) is a special case of the Gale-Robinson recurrence (\ref{887}) for $\left(N, a, c\right)=\left(5, 1, 2\right)$.

The classical integer sequence generated by (\ref{889987}) and   listed in the OEIS as \textrm{A006721}, is defined by the initial values $t_j=1$ for $j=0,\ldots, 4$ and $\left(\alpha, \beta\right)=(1, 1)$. This  sequence surprisingly combines combinatorial, arithmetic, and geometric properties about which we would like to say a few words. 
For example, in \cite{Davis}, one can find the following lemma, which relates \textrm{A006721}  to the elliptic curve $E: y^2+xy=x\left(x-1\right)\left(x+2\right)$.  
\begin{lemma}
The relation between \textrm{A006721} and the elliptic curve $E$ is determined by the relation
\[
Q+nP=\left(\frac{t_{n+2}^2-t_nt_{n+4}}{t_{n+2}^2}, \frac{4 t_{n}t_{n+2}t_{n+4}-t_n^2t_{n+6}-t^3_{n+2}}{t_{n+2}^3}\right),
\]
where a point $P=(2, 2)$ of infinite order and a point $Q=(0, 0)$ of order two on the curve.
\end{lemma}
Thus, the sequence encodes the complex arithmetic of adding points on the elliptic curve. The connection of this elliptic curve with the sequence \textrm{A006721}, as far as we know, was found by Elkies. 

In the 1990s, Buchholz and Rathbun, in contrast to Schubert's erroneous statement \cite{Dickson}, discovered an infinite family of Heronian triangles with two rational medians \cite{Buchholz}. They noticed that this infinite family is somehow connected to \textrm{A006721} and some associated sequence. Hone went further and provided explicit formulas expressing the sides of the triangles, the two rational medians, and their area directly in terms of the terms of two intertwined Somos-5 sequences \cite{Hone3}. The significance of these works is that they established that  the parameters of the triangles in this infinite family \textit{live} on corresponding elliptic curve. An infinite  sequence of such Heronian triangles correspond exactly to the sequence of points $Q+nP$ on the elliptic curve.


The aim of this work is to present and investigate an infinite family of homogeneous equations of the form 
\begin{equation}
t_nt_{n+2g+3}=\frac{\sum_{j=0}^{g}\alpha_j B^j_{2g-j}(n+1)}{\prod_{j=3}^{2g} t_{n+j}},\; g\geq 1, 
\label{8878886500097}
\end{equation}
which we denote as $R_{2g+3}$.  The coefficients $\alpha_j$ in (\ref{8878886500097}) are supposed to be arbitrary.
On the right-hand side of (\ref{8878886500097}), it is necessary to define the homogeneous discrete polynomials $B^k_s(n)$. In the next section, we will show how these polynomials are successively defined using a certain recurrence relation,  but for now we only note that $B^j_{2g-j}(n)$ is a homogeneous polynomial in the variables $\left(t_n,\ldots, t_{n+2g+1}\right)$. The degree of this polynomial is $2g$. Thus, (\ref{8878886500097}) is supposed to be a homogeneous recurrence relation of degree $2g$ for $g\geq 1$. From what we have said, it follows that (\ref{8878886500097}) is a relation of the form (\ref{887777}), where the right-hand side is a linear combination of certain Laurent polynomials.
Only in one case, namely $g=1$, should the product on the left-hand side of (\ref{8878886500097}) be considered equal to one. In this case, (\ref{8878886500097}) reduces to Somos-5 given by  (\ref{889987}) with $\left(\alpha, \beta\right)=\left(\alpha_1, \alpha_0\right)$.

In what follows, we will show that (\ref{8878886500097}) are homogeneous recurrences of the form (\ref{887777}). For illustration, let us  write explicitly the recurrence relation (\ref{8878886500097}) for the case $g=2$. It has the following form:
\begin{equation}
t_nt_{n+7}=\alpha_0 t_{n+2}t_{n+5}+\alpha_1 \frac{t_{n+1}t_{n+4}^2t_{n+5}+t_{n+2}^2t_{n+5}^2+t_{n+2}t_{n+3}^2t_{n+6}}{t_{n+3}t_{n+4}}+\alpha_2 t_{n+1}t_{n+6}.
\label{6887638}
\end{equation}
It is obvious that this recurrence is of the form (\ref{887777}).   Let us put $n=0$, so that corresponding  $L$ becomes a Laurent polynomial in the variables $\left(t_{1},\ldots, t_{6}\right)$. This polynomial  exhibits a symmetry: $L(t_{1},\ldots, t_{6}) = L(t_{6},\ldots, t_{1})$. In other words, renumbering the arguments in reverse order yields the same polynomial. This example suggest that this symmetry is likely a specific property not only of all recurrences in our family, but also of many other recurrences with the Laurent property found in the literature. Naturally, for this to hold, all the discrete polynomials $B^k_s(n)$ appearing in the definition of the recurrences $R_{2g+3}$ must themselves possess this symmetry. 
We also observe that  under the condition $\alpha_1=0$ this recurrence  reduces to the Gale-Robinson one (\ref{887}) for the case $\left(N, a, c\right)=\left(7, 1, 2\right)$.

The following more general statement holds.
\begin{proposition} \label{6754321}
Under the condition $\alpha_j=0,\; \forall j=1,\ldots, g-1$ the recurrence $R_{2g+3}$ reduces to the Gale-Robinson one (\ref{887}) of type $\left(N, a, c\right)=\left(2g+3, 1, 2\right)$. 
\end{proposition}

Let us also consider the following question: what is the number of Laurent monomials occurring in  $R_{2g+3}$? The answer to this question is quite interesting and given by the following statement.
\begin{proposition}   \label{00098767}
The number of monomials in relation (\ref{8878886500097}) is equal to $F_{2g+1}+1$, where $F_k$ denotes the $k$-th Fibonacci number.
\end{proposition}

Let us formulate another property of recurrences (\ref{8878886500097}), which is a generalization of our observation for the case $g=2$ (and for $g=1$, of course).
\begin{proposition}   \label{004432767}
Setting $n=0$ in (\ref{8878886500097}), the right-hand side of this relation admits the symmetry $L(t_{1},\ldots, t_{2g+2}) = L(t_{2g+2},\ldots, t_{1})$.
\end{proposition}

Also we will show the following simple property for the $R_{2g+3}$.
\begin{proposition}   \label{00444327}
Given a pair of nonzero constants $A$ and $B$ and any solution of the $R_{2g+3}$. The transformation  $t_n\mapsto AB^n t_n$ gives another solution to this recurrence.
\end{proposition}

Let us now formulate the main property of the recurrences $R_{2g+3}$ that we are going to prove.
\begin{theorem}  \label{76888766}
Any recurrence $R_{2g+3}$ has the Laurent property. More precisely, for all $n\in\mathbb{Z}$, $t_n$ belongs to the ring $\mathbb{Z}[t_0^{\pm 1},\ldots, t_{2g+2}^{\pm 1}; \alpha_0,\ldots,\alpha_g]$.
\end{theorem}

The plan of the article is as follows. In the following section, we will write down and examine  recurrence that is directly related to (\ref{8878886500097}). It will be defined using some discrete polynomials $T^k_s(n)$.  We will show how our  recurrence $R_{2g+3}$ under study is related to this associated recurrence by means of a certain substitution.  The point of dealing with the associated recurrence, rather than with (\ref{8878886500097}) directly, is that it  is easier to work with.  In Section \ref{765439876}, we construct a Lax representation for associated recurrence. This representation has its own specific features. In particular, it is closely related to a certain continued fraction $St_n(x)$. Essentially we need this Lax representation to prepare the ground for the proof of Theorem \ref{76888766}.  Section \ref{56432908} is essentially devoted to the proof of Theorem \ref{76888766}, but not only that.
There we establish a connection between (\ref{8878886500097}) and the discrete Mumford dynamical system  introduced in \cite{Hone1}. Short Section \ref{7654098756} provides background information about integer sequences generated by  (\ref{8878886500097}) in a some particular case. Finally, Section \ref{786453209} is devoted to remarks and conjectures.

\begin{remark}
We would like to use the two terms discrete equation and recurrence relation together  and consider them as synonyms. We also consider the terms sequence and solution of a discrete equation as synonymous.
\end{remark}

\section{Associated recurrence}

\subsection{Definition}

The first thing we want to do is to fully define recurrences of the form (\ref{8878886500097}).
To this aim, for any integer $g \geq 1$, we define the associated recurrence 
\begin{equation}
\prod_{j=0}^{2g} u_{n+j}=\sum_{j=0}^g \alpha_j T^j_{2g-j}(n+1).
\label{6766438}
\end{equation}
Here it is assumed that $T^k_s(n)$, for $k, s \geq 1$, is some discrete homogeneous polynomial to be determined. We should immediately explain why we have written this relation here.
We will show that by using the substitution
\begin{equation}
u_n=\frac{t_n t_{n+3}}{t_{n+1}t_{n+2}}
\label{676678888768}
\end{equation}
in (\ref{6766438}), we can reduce it to  (\ref{8878886500097}).   
\begin{remark}
Some comments are in order. Note that here we are working backwards, so to speak, compared to \cite{Hone2}. Having noted the property of the Somos-5  mentioned in Proposition 1, Hone goes on to say that it is natural to define \textit{gauge-invariant} combination  $w_n=t_nt_{n+2}/t^2_{n+1}$ and work with the recurrence for the sequence $\left(w_n\right)_{n\in \mathbb{Z}}$. Then, he constructs an invariant and a general solution for this recurrence. Note that one has $u_n=w_nw_{n+1}$, in (\ref{676678888768}).  Note that Proposition \ref{00444327} can be regarded as proven simply by the construction of our recurrences (\ref{8878886500097}). Equation (\ref{6766438}) is, as we intend, a consequence of $R_{2g+3}$ for any $g\geq 1$. If Theorem \ref{76888766} holds, then one can say that the substitution (\ref{676678888768}) yields a \textit{Laurentification} of  (\ref{6766438}).
\end{remark}

\subsection{Discrete polynomials $T^k_s(n)$}

In what follows, we largely follow the works \cite{Svinin3} and \cite{Svinin4} . 
From a practical point of view, a good definition of $T^k_s(n)$ is given by the recurrence relation:
\begin{equation}
T^k_s(n)=\sum_{j=0}^{s-k} u_{n+j} T^{k-1}_{s-j-1}(n+j+2),\; \forall s\geq k.
\label{676678}
\end{equation}
For example, one can start with $T^0_s(n)=1$ and then step by step compute the polynomials $T^k_s(n)$ of increasingly higher degree $k$.

We can also write $T^k_s(n)$ in explicit form, namely as
\begin{equation}
T^k_s(n)=\sum_{0\leq \lambda_1<\cdots< \lambda_k\leq s-1} \prod_{j=1}^k u_{n+\lambda_j+j-1}.
\label{67999998}
\end{equation}
Expression (\ref{67999998}) represents a summation over a finite set of admissible values of $\lambda_j$. Let us denote $D_{k, s}=\left\{\lambda_j : 0\leq \lambda_1<\cdots< \lambda_k\leq s-1\right\}$. Note that in the case $s=1,\ldots, k-1$, we have $D_{k, s}=\varnothing$, and consequently $T^k_s(n)=0$ for these values of $s$. In turn, if $s=k$, then the set $D_{k, s}$ consists of a single element $\left\{\lambda_j=j,\;\; j=0,\ldots, k-1\right\}$, and thus we obtain
\begin{equation}
T^k_k(n)=\prod_{j=0}^{k-1}u_{n+2j},\;\; \forall k\geq 1.
\label{6799999000988}
\end{equation}

The drawback of this definition is that (\ref{67999998}) is rather difficult to use in practice compared to (\ref{676678}); nevertheless, as we will see below, some useful facts can be derived from it. Therefore, it is better to use these definitions together. What can we observe by looking at (\ref{67999998})? We see that $T^k_s(n)$ is a homogeneous polynomial in the variables $\left(u_n,\ldots, u_{n+k+s-2}\right)$ of degree $k$. Of course, it is not obvious that (\ref{676678}) and (\ref{67999998}) yield the same infinite class of discrete polynomials; however, as we will show below in the further exposition, we will demonstrate this fact.

\subsection{Discrete polynomials $B^k_s(n)$ and definition of recurrences (\ref{8878886500097})}

The first thing to investigate is how to go from (\ref{6766438}) to  (\ref{8878886500097}). To do this, we need to define the mapping $T^k_s(n)\mapsto B^k_s(n)$. Let us write
\begin{equation}
B^k_s(n)=T^k_s(n)\prod_{j=1}^{k+s} t_{n+j},\;\; \forall s\geq k,\; k\geq 1,
\label{65400000987}
\end{equation}
provided that we have made the substitution (\ref{676678888768}) into $T^k_s(n)$. But we must be sure that $B^k_s(n)$ are indeed polynomials. The following lemma is useful for proving this.
\begin{lemma} \label{65439876} \cite{Svinin4}
The two relations
\begin{eqnarray}
T^k_s(n)&=&T^k_{s-1}(n)+u_{n+s+k-2}T^{k-1}_{s-1}(n) \label{6540987} \\
&=&T^k_{s-1}(n+1)+u_nT^{k-1}_{s-1}(n+2)
\label{67776587}
\end{eqnarray}
are identities.
\end{lemma}
\begin{proof}
The relations (\ref{6540987}) and (\ref{67776587}) arise as a result of a suitable partition of the set $D_{k, s}$ into two disjoint sets. For example, one can verify that (\ref{6540987}) corresponds to the partition $D_{k, s}=D^{(1)}_{k, s}\sqcup D^{(2)}_{k, s}$ with
\[
D^{(1)}_{k, s}=\left\{\lambda_j : \lambda_1=0,\;\; 1\leq \lambda_2<\cdots< \lambda_k\leq s-1\right\}
\]
and 
\[
D^{(2)}_{k, s}=\left\{\lambda_j : 1\leq \lambda_1<\cdots< \lambda_k\leq s-1\right\}.
\]
\end{proof}

Note that (\ref{676678}) can be obtained by successive applications of (\ref{67776587}). From this, in turn, follows the equivalence of the two different definitions (\ref{676678}) and (\ref{67999998}) for the polynomials $T^k_s(n)$.

From Lemma \ref{65439876} together with (\ref{65400000987}), we obtain the following statement.
\begin{lemma}  \label{7654876}
Discrete polynomials $B^k_s(n)$ can be computed using one of the following two recurrence relations:
\begin{eqnarray}
B^k_s(n)&=&t_{n+k+s}B^k_{s-1}(n)+t_{n+k+s-2}t_{n+k+s+1}B^{k-1}_{s-1}(n) \label{654443287} \\
&=&t_{n+1}B^k_{s-1}(n+1)+t_nt_{n+3}B^{k-1}_{s-1}(n+2).
\label{65445554337}
\end{eqnarray}
\end{lemma}
However, for these recurrence relations to work, additional data are needed. Taking (\ref{6799999000988}) into account, we can write:
\[
B^k_k(n)=\prod_{j=0}^{k-1} t_{n+2j} \prod_{j=0}^{k-1} t_{n+2j+3},\;\; \forall k\geq 1.
\]
Fortunately, this polynomial can be computed without using any recurrence relations. We also need to define $B^0_s(n)=\prod_{j=1}^s t_{n+j}, \forall s\geq 1$. Note that the latter is consistent with (\ref{65400000987}) provided that we set $T^0_s(n)=1, \forall s\geq 1$. Now we have enough data to use relation (\ref{654443287}) or, if desired, relation (\ref{65445554337}) to determine the polynomials $B^k_s(n)$, first for $k=1$, then for $k=2$, and so on.

Obviously, $\deg\left(B^0_s\right)=s$ and $\deg\left(B^k_k\right)=2k$. Moreover, the two terms on the right-hand sides of (\ref{654443287}) and (\ref{65445554337}) have the same degree, equal to $k+s$. All this proves that a discrete polynomials $B^k_s(n)$, computed using (\ref{654443287}) (or (\ref{65445554337})), is indeed a homogeneous one of degree $k+s$ in the variables $\left(t_n,\ldots, t_{n+k+s+1}\right)$.

As a complement to Lemma \ref{7654876}, we prove the following statement, which generalizes our previous observation. For convenience, we take $n=0$, so that a discrete polynomial $B^k_s(n)$ turns into a polynomial $B^k_s(0)$ in the variables $\left(t_0,\ldots, t_{k+s+1}\right)$.      For convenience, let us denote this polynomial simply  by $B^k_s$.   
\begin{lemma}
The polynomial  $B^k_s$ exhibits a symmetry: 
\[
B^k_s\left(t_{k+s+1},\ldots, t_{0}\right)=B^k_s\left(t_0,\ldots, t_{k+s+1}\right),\;\; \forall k\geq 0, s\geq 1,
\]
that is, renumbering the arguments $\left(t_0,\ldots, t_{k+s+1}\right)$ in reverse order yields the same polynomial.
\end{lemma}
\begin{proof}
From (\ref{654443287}) it follows the recurrent relation
\begin{equation}
B^k_s\left(t_0,\ldots, t_{k+s+1}\right)=t_{k+s}B^k_{s-1}\left(t_0,\ldots, t_{k+s}\right)+t_{k+s-2}t_{k+s+1}B^{k-1}_{s-1}\left(t_0,\ldots, t_{k+s-1}\right).
\label{867432098}
\end{equation}
Given any values of $\left(k, s\right)$, we perform the substitution 
\[
t_j\rightarrow t_{k+s+1-j},\; j=0,\ldots, k+s+1
\]
in (\ref{867432098}). It thereby turns into
\[
B^k_s\left(t_{k+s+1},\ldots, t_{0}\right)=t_{1}B^k_{s-1}\left(t_{k+s+1},\ldots, t_{1}\right)+t_{3}t_{0}B^{k-1}_{s-1}\left(t_{k+s+1},\ldots, t_{2}\right).
\]
Suppose now  that the statement of the lemma has already been proved for $B^k_{s-1}$ and $B^{k-1}_{s-1}$. Then, by virtue of (\ref{65445554337}), we obtain this property for $B^k_{s}$. It is now clear that the lemma admits a proof by induction. Indeed, since $B^0_1=t_1$ and $B^1_1=t_0t_3$ satisfy the statement of the lemma, it follows by induction that this statement holds for all polynomials $B^1_s,\; s\geq 2$. Next, by induction, the statement of the lemma is proved for all $B^2_s$, and so on and so forth.
\end{proof}
Notice that having proved this lemma, we have thereby proved Proposition \ref{004432767}.

Let us now look at what recurrence relation we obtain as a result of substituting (\ref{676678888768}) into (\ref{6766438}). We get:
\begin{equation}
\prod_{j=3}^{2g} t_{n+j}\cdot t_nt_{n+2g+3}=\sum_{j=0}^g \alpha_j B^j_{2g-j}(n+1).
\label{564000032987}
\end{equation}
In addition, note that 
\[
B^0_{2g}(n+1)=\prod_{j=3}^{2g} t_{n+j}\cdot t_{n+2}t_{n+2g+1}\; \mbox{and}\; B^g_g(n+1)=\prod_{j=3}^{2g} t_{n+j}\cdot t_{n+1}t_{n+2g+2}.
\] 
Taking into account these expressions, we can rewrite (\ref{564000032987}) as
\begin{equation}
\prod_{j=3}^{2g} t_{n+j}\cdot \left(t_nt_{n+2g+3}-\alpha_g t_{n+1}t_{n+2g+2}-\alpha_0 t_{n+2}t_{n+2g+1}\right)=\sum_{j=1}^{g-1}\alpha_j B^j_{2g-j}(n+1)
\label{5640087}
\end{equation}
Rewriting (\ref{564000032987}) in the form (\ref{5640087})  gives a proof of Proposition \ref{6754321}.
\begin{lemma} \label{9865987}
The number of monomials on the right-hand side of relation (\ref{6766438}) is equal to $F_{2g+1}$.
\end{lemma}
\begin{proof}
Obviously, $\# D_{1, s}=s$. From (\ref{676678}) and the binomial identity
\[
\left(\begin{array}{c} s\\ k \end{array}\right)=\sum_{j=0}^{s-k} \left(\begin{array}{c} s-j-1\\ k-1 \end{array}\right),
\]
by induction we obtain that the number of monomials occurring in $T^k_s(n)$ is equal to the binomial coefficient $\binom{s}{k}$, or in other words, $\# D_{k, s}=\binom{s}{k}$. In turn, using the well-known formula for Fibonacci numbers
\[
F_{2g+1}=\sum_{j=0}^g \binom{2g-j}{j},
\]
we obtain that the number of monomials on the right-hand side of (\ref{6766438}) is equal to $F_{2g+1}$.
\end{proof}
According to Lemma \ref{9865987}, the number of monomials on the right-hand side of (\ref{564000032987}) is equal to $F_{2g+1}$. Thus, the total number of monomials defining (\ref{564000032987}) is $F_{2g+1}+1$, and therefore we have proved Proposition \ref{00098767}.

\section{Lax representation for equation (\ref{6766438})}   \label{765439876}

The purpose of this section is to show the Lax representation for recurrence (\ref{6766438}). Since discrete polynomials $T^k_s(n)$ are involved in the definition of the recurrence (\ref{6766438}), it is natural to expect that they will somehow appear in the Lax representation.

\subsection{Lax pair}

In this section, we construct a Lax representation for the associated recurrence (\ref{6766438}). To this end, consider a pair of matrices
\[
L_n(x)=
\left(
\begin{array}{cc}
P_n(x) & R_n(x) \\
Q_n(x) & -P_n(x) 
\end{array}
\right)\; \mbox{and}\;\;
M_n(x)= 
\left(
\begin{array}{cc}
1 & u_n x \\
1 & 0 
\end{array}
\right)
\]
and a relation of the form
\begin{equation}
L_n M_n=M_n  L_{n+1},
\label{778}
\end{equation}
where $\left(P_n, Q_n, R_n\right)_{n\in\mathbb{Z}}$ is a sequence of triples of polynomials
\begin{equation}
P_n(x)=1+\sum_{j= 1}^g p_{j, n}x^j,\;\; Q_n(x)=2+\sum_{j= 1}^g q_{j, n}x^j\;\; \mbox{and}\;\;  R_n(x)=\sum_{j= 1}^{g+1} r_{j, n}x^j.
\label{79998}
\end{equation}
We want to define the coefficients of these polynomials in such a way that relation (\ref{778}) becomes generating one for (\ref{6766438}).
\begin{theorem}   \label{8975465}
Equation (\ref{6766438}) admits a Lax representation (\ref{778}), with the coefficients of polynomials (\ref{79998}) of the form
\begin{equation}
q_{k, n}=2T^k_{2g-k}(n+1)+2\sum_{j=1}^k T^{k-j}_{2g-k-j}(n+1)\frac{T_{j, g-j}(n+1)}{\prod_{s=0}^{2g-2j} u_{n+s}},
\label{7766548888}
\end{equation}
\begin{equation}
r_{k, n}=2\sum_{j=0}^{k-1} T^{k-j-1}_{2g-k-j+1}(n)\frac{T_{j, g-j}(n)}{\prod_{s=0}^{2g-2j-1} u_{n+s}},
\label{7766548}
\end{equation}
and
\begin{eqnarray}
p_{k, n}&=&T^k_{2g-k+1}(n)+\sum_{j=1}^k T^{k-j}_{2g-k-j}(n+1)\frac{T_{j, g-j}(n+1)}{\prod_{s=0}^{2g-2j} u_{n+s}} \nonumber\\
&& -\sum_{j=0}^{k-1} T^{k-j-1}_{2g-k-j+1}(n)\frac{T_{j, g-j}(n)}{\prod_{s=0}^{2g-2j-1} u_{n+s}} +T^{k-2}_{2g-k}(n+2)\frac{T_{0, g}(n+1)}{\prod_{j=1}^{2g-1} u_{n+j}} \nonumber\\
&&+u_n\sum_{j=1}^{k-1} T^{k-j-1}_{2g-k-j+1}(n+2)\frac{T_{j, g-j}(n+2)}{\prod_{s=1}^{2g-2j+1} u_{n+s}}.
\label{5643987}
\end{eqnarray}
\end{theorem}

We will prove Theorem \ref{8975465} later, but for now let us  discuss some details.
The formulas in this theorem for the coefficients of the polynomials $\left(P_n, Q_n, R_n\right)$ look rather cumbersome. To avoid excessive cumbersomeness, we use the following notation:
\begin{equation}
T_{s, k}(n)=\sum_{j=0}^k \alpha_{s+j}T^{j}_{2k-j}(n).
\label{786543}
\end{equation}
For example, 
\[
T_{s, 0}(n)=\alpha_{s},\;\; T_{s, 1}(n)=\alpha_{s}+\alpha_{s+1}T^1_1(n),
\]
\[
T_{s, 2}(n)=\alpha_{s}+\alpha_{s+1}T^1_3(n) +\alpha_{s+2}T^2_2(n),\ldots
\]
One can verify that the coefficients of the polynomials $\left(P_n, Q_n, R_n\right)$ given in the theorem are Laurent polynomials in the variables $\left(u_n,\ldots, u_{n+2g-1}\right)$, and   depend linearly, albeit inhomogeneously, on the parameters $\left(\alpha_0,\ldots, \alpha_g\right)$ entering into  (\ref{6766438}). More precisely, the coefficients of the polynomials $\left(P_n, Q_n, R_n\right)$ are linear combinations of discrete Laurent polynomials. The coefficients $q_{k, n}$ and $p_{k, n}$ contain a \textit{free} term, that is, a term that is independent of $\alpha_k$.

\subsection{Proof of Theorem \ref{8975465}}

The proof of this theorem is technical and and possibly is not of sufficient interest.  

First of all we note that our associated recurrence  (\ref{6766438}), by virtue of the adopted notation (\ref{786543}), can be written as
\begin{equation}
\prod_{j=0}^{2g} u_{n+j}=T_{0, g}(n+1).
\label{77000988}
\end{equation}

To verify that, for (\ref{778}) to hold, it is necessary that the sequence of triples of polynomials $\left(P_n, Q_n, R_n\right)_{n\in\mathbb{Z}}$ satisfy the relations: 
\begin{equation}
P_n+P_{n+1}=Q_n,\;\; 2 P_n - Q_n + R_n=u_n x Q_{n+1}\;\; \mbox{and}\;\; R_{n+1}=u_n x Q_n.
\label{770098}
\end{equation}
Componentwise, they are written as:
\begin{equation}
p_{k, n+1}=q_{k, n}-p_{k, n},\;\; q_{k, n+1}=\frac{2 p_{k+1, n}-q_{k+1, n}+r_{k+1, n}}{u_n}\;\;\mbox{and}\;\; r_{k, n+1}= u_{n}q_{k-1, n}.
\label{770099990988}
\end{equation}
Here, in the third formula, in the case $k=1$, we  set $q_{0, n}=2$. 

Let  us proceed step by step. First, we verify the third relation in (\ref{770099990988}). Taking into account  (\ref{77000988}), one has
\begin{eqnarray*}
r_{k, n+1}&=&2T^{k-1}_{2g-k+1}(n+1)\frac{T_{0, g}(n+1)}{\prod_{s=1}^{2g} u_{n+s}}+2\sum_{j=1}^{k-1} T^{k-j-1}_{2g-k-j+1}(n+1)\frac{T_{j, g-j}(n+1)}{\prod_{s=1}^{2g-2j} u_{n+s}} \\
&=&2  u_n T^{k-1}_{2g-k+1}(n+1)+2u_n\sum_{j=1}^{k-1} T^{k-j-1}_{2g-k-j+1}(n+1)\frac{T_{j, g-j}(n+1)}{\prod_{s=0}^{2g-2j} u_{n+s}} \\
&=&u_n q_{k-1, n}.
\end{eqnarray*}
Thus this relation holds.

Further, we compute the coefficient $p_{k, n}$ so that the second relation in (\ref{770099990988}) is automatically satisfied, that is,
\begin{eqnarray*}
p_{k, n}&=&\frac{q_{k, n}}{2}-\frac{r_{k, n}}{2} + u_n\frac{q_{k-1, n+1}}{2}   \\
&=&T^k_{2g-k}(n+1)+\sum_{j=1}^k T^{k-j}_{2g-k-j}(n+1)\frac{T_{j, g-j}(n+1)}{\prod_{s=0}^{2g-2j} u_{n+s}} \\
&& -\sum_{j=0}^{k-1} T^{k-j-1}_{2g-k-j+1}(n)\frac{T_{j, g-j}(n)}{\prod_{s=0}^{2g-2j-1} u_{n+s}}  \\
&& +u_nT^{k-1}_{2g-k+1}(n+2)+u_n\sum_{j=1}^{k-1} T^{k-j-1}_{2g-k-j+1}(n+2)\frac{T_{j, g-j}(n+2)}{\prod_{s=1}^{2g-2j+1} u_{n+s}}.
\end{eqnarray*}
We  need here to exclude $u_{n+2g}$  with the help  of (\ref{77000988}). It is contained only in  $u_nT^{k-1}_{2g-k+1}(n+2)$.
To express the dependence on $u_{n+2g}$ more explicitly, we must use the identity
\[
T^{k-1}_{2g-k+1}(n+2)=T^{k-1}_{2g-k}(n+2) + u_{n+2g} T^{k-2}_{2g-k}(n+2)
\]
that follows from (\ref{6540987}).  Besides, we use (\ref{77000988}) to eliminate $u_{n+2g}$. Then
\begin{eqnarray*}
u_nT^{k-1}_{2g-k+1}(n+2)&=&u_nT^{k-1}_{2g-k}(n+2)+u_nu_{n+2g}T^{k-2}_{2g-k}(n+2) \\
&=&u_nT^{k-1}_{2g-k}(n+2)+T^{k-2}_{2g-k}(n+2)\frac{T_{0, g}(n+1)}{\prod_{j=1}^{2g-1} u_{n+j}}.
\end{eqnarray*}
Now it remains to use identity 
\[
T^k_{2g-k+1}(n)=T^k_{2g-k}(n+1)+u_nT^{k-1}_{2g-k}(n+2)
\]
to  obtain the expression (\ref{5643987}) we need.

Let us define the following expression:
\begin{eqnarray}
o_{k, s}(n)&=&T_{k, s}(n)-T_{k, s}(n+1)+u_{n+2s-1} T_{k+1, s-1}(n+1) \nonumber\\
&&-u_nT_{k+1, s-1}(n+2),\;\;\forall s\geq 0,
\label{76543}
\end{eqnarray}
where $T_{k, s}(n)$ is linear combination  (\ref{786543}).
\begin{lemma} \label{67543987}
The relation $o_{k, s}(n)=0$ is an identity.
\end{lemma}
\begin{proof}
This relation is a generating one for a finite set of identities of the form
\[
T^k_{2s-k}(n)+u_{n+2s-1} T^{k-1}_{2s-k-1}(n+1)=T^k_{2s-k}(n+1)+u_{n} T^{k-1}_{2s-k-1}(n+2),
\]
and they in turn are a special case of identities of the form (\ref{6540987})=(\ref{67776587}).
\end{proof}

To finally prove the theorem, it is necessary to verify the first relation in (\ref{770099990988}). 
Let us denote
\begin{eqnarray}
\mathcal{O}_{g, k}&=&\prod_{j=0}^{2g} u_{n+j}\left(p_{k, n}+p_{k, n+1}-q_{k, n}\right) 
-\left(u_n T^{k-1}_{2g-k}(n+2)  +u_{n+2g} T^{k-1}_{2g-k}(n+1)\right) \nonumber \\
&&\times\left(\prod_{j=0}^{2g} u_{n+j}-T_{0, g}(n+1)\right).
\label{76589654}
\end{eqnarray}
First of all, note that in (\ref{76589654}) the \textit{free} term is absent due to the identity 
\begin{eqnarray*}
&&T^k_{2g-k+1}(n)+T^k_{2g-k+1}(n+1)-2T^k_{2g-k}(n+1) \\
&&\;\;\;\;\;\;\;\;\;\;\;\;\;\;\;\;\;-u_n T^{k-1}_{2g-k}(n+2)-u_{n+2g} T^{k-1}_{2g-k}(n+1)=0.
\end{eqnarray*}
The last relation can be obtained by combining identities (\ref{6540987} ) and (\ref{67776587}).
It can be verified  that all $\mathcal{O}_{g, k}$ are expressed linearly and homogeneously through  combinations  $T_{k, g-k}(n)$ only in connectives $o_{k, g-k}(n)$ defined by (\ref{76543}). For example,  $\mathcal{O}_{g, 1}=u_{n+2g}o_{0, g}(n)$. Due to  Lemma \ref{67543987}, we get that $\mathcal{O}_{g, k}=0$.  And this in turn means that the first relation in (\ref{770099990988}) holds provided  (\ref{77000988}) is satisfied.

\subsection{Example of Lax representation. Somos-5}

Let us consider the simplest case $g=1$.  According to  Theorem \ref{8975465}, for the recurrence 
\[
u_nu_{n+1}u_{n+2}=\alpha_0+\alpha_1 u_{n+1},
\]
the coefficients of the polynomials (\ref{79998}) are as follows:  
\begin{eqnarray}
q_{1, n}&= &2u_{n+1}+2\frac{\alpha_1}{u_n} \nonumber\\
&= &2\frac{t_{n+1}t_{n+4}}{t_{n+2}t_{n+3}}+2\alpha_1 \frac{t_{n+1}t_{n+2}}{t_{n}t_{n+3}},
\label{8754309}
\end{eqnarray}
\begin{eqnarray}
p_{1, n}&=& u_n+u_{n+1}+\frac{\alpha_1}{u_n}-\frac{\alpha_0+u_n\alpha_1}{u_nu_{n+1}} \nonumber\\
&=& \frac{t_{n}t_{n+3}}{t_{n+1}t_{n+2}}+\frac{t_{n+1}t_{n+4}}{t_{n+2}t_{n+3}}+\alpha_1 \left(\frac{t_{n+1}t_{n+2}}{t_{n}t_{n+3}}-\frac{t_{n+2}t_{n+3}}{t_{n+1}t_{n+4}}\right)-\alpha_0 \frac{t^2_{n+2}}{t_{n}t_{n+4}},
\label{875439986509} 
\end{eqnarray}
\begin{eqnarray}
r_{1, n}&=&2\frac{\alpha_0+u_n\alpha_1}{u_nu_{n+1}} \nonumber\\
&=&2\alpha_0 \frac{t^2_{n+2}}{t_{n}t_{n+4}}+2\alpha_1 \frac{t_{n+2}t_{n+3}}{t_{n+1}t_{n+4}}
\label{875476486509777} 
\end{eqnarray}
and
\begin{eqnarray}
r_{2, n}&=&2\frac{\alpha_0+u_n\alpha_1}{u_{n+1}}+2\alpha_1 \nonumber\\
&=&2\alpha_0 \frac{t_{n+2}t_{n+3}}{t_{n+1}t_{n+4}}+ 2\alpha_1\left(\frac{t_nt^2_{n+3}}{t^2_{n+1}t_{n+4}}+1\right).
\label{875476486509} 
\end{eqnarray}
Notice that we made the substitution  (\ref{676678888768}) here.
After that we obtained  expressions (\ref{8754309}),  (\ref{875439986509}),  (\ref{875476486509777}) and (\ref{875476486509}) which  define a Lax pair for the  Somos-5 given by (\ref{889987}). One can directly verify that each of the relations 
\[
p_{1, n+1}=q_{1, n}-p_{1, n},\;\; q_{1, n+1}=\frac{r_{2, n}}{u_n},\;\; r_{1, n+1}= 2 u_{n}\;\;\mbox{and}\;\; r_{2, n+1}= u_{n}q_{1, n}
\]
is equivalent to Somos-5. Observe that 
\[
\frac{q_{1, 0}}{2}=\frac{t_{1}t_{4}}{t_{2}t_{3}}+\alpha_1 \frac{t_{1}t_{2}}{t_{0}t_{3}}\;\;\mbox{and} \;\;
p_{1, 0}-q_{1, 0}=\frac{t_{0}t_{3}}{t_{1}t_{2}}-\frac{t_{1}t_{4}}{t_{2}t_{3}}-\alpha_1 \left(\frac{t_{1}t_{2}}{t_{0}t_{3}}+\frac{t_{2}t_{3}}{t_{1}t_{4}}\right)-\alpha_0 \frac{t^2_{2}}{t_{0}t_{4}}
\]
belong to the ring $\mathcal{L}_5=\mathbb{Z}[t_0^{\pm 1}, t_1^{\pm 1}, t_2^{\pm 1}, t_3^{\pm 1}, t_4^{\pm 1}; \alpha_0, \alpha_1]$.

\begin{remark}
Of course, using Theorem \ref{8975465}, one could give examples of Lax pairs for $R_7$, $R_9$, and so on, but the corresponding formulas are extremely cumbersome and there is no point in presenting them.
\end{remark}

\subsection{Consequences of Theorem \ref{8975465}}

Notice that $q_{k, 0}$ and $p_{k, 0}$ as Laurent polynomials of the arguments  $\left(u_0,\ldots,  u_{2g-1}\right)$ are given to us explicitly thanks to Theorem \ref{8975465}. 
If we make a substitution  (\ref{676678888768}) into these functions,  they  become Laurent polynomials of $\left(t_0,\ldots,  t_{2g+3}\right)$.
The following lemma is important for the proof of  Theorem \ref{76888766}, which we will prove in the next section.
\begin{lemma} \label{6754329876} 
Functions $q_{k, 0}/2$ and $p_{k, 0}-q_{k, 0}$ belong to the ring 
\[
\mathcal{L}_{2g+3}=\mathbb{Z}[t_0^{\pm 1},\ldots,  t_{2g+2}^{\pm 1}; \alpha_0,\ldots, \alpha_g],
\]  
for all $k=1,\ldots, g$ and $g\geq 1$.
\end{lemma}
\begin{proof}
By direct inspection.  One has that $q_{k, 0}/2$ and $p_{k, 0}-q_{k, 0}$ belong to the ring $\mathbb{Z}[u_0^{\pm 1},\ldots,  u_{2g-1}^{\pm 1}; \alpha_0,\ldots, \alpha_g]$ for any $k=1,\ldots, g$. From this the lemma follows.
\end{proof}
The following lemma is useful for establishing the relationship between the associated recurrence (\ref{6766438}) and the so-called discrete Mumford dynamical system \cite{Hone1}.
\begin{lemma} \label{6753290}
Given any $g\geq 1$, by (\ref{7766548888}), (\ref{7766548}) and (\ref{5643987}), one has 
\[
p_{1,n}-\frac{q_{1,n}}{2}+\frac{r_{1,n}}{2}=u_n.
\]
\end{lemma}
\begin{proof}
For clarity, let us simply write down
\[
q_{1, n}=2T^1_{2g-1}(n+1)+2\frac{T_{1, g-1}(n+1)}{\prod_{s=0}^{2g-2} u_{n+s}},\;\; r_{1, n}=2\frac{T_{0, g}(n)}{\prod_{s=0}^{2g-1} u_{n+s}}
\]
and
\[
p_{1, n}=T^1_{2g}(n)+\frac{T_{1, g-1}(n+1)}{\prod_{s=0}^{2g-2} u_{n+s}}-\frac{T_{0, g}(n)}{\prod_{s=0}^{2g-1} u_{n+s}},
\]
after which the assertion of the lemma is easily verified by inspection.
\end{proof}

\section{Discrete Mumford dynamical system} \label{56432908}

\subsection{Definition}

Suppose now that the coefficients of the polynomials $\left(P_n, Q_n, R_n\right)$ do not depend on $\left(u_n,\ldots, u_{n+2g-1}\right)$ and $\left(\alpha_0,\ldots, \alpha_g\right)$ like in Theorem \ref{8975465} and generally do not depend on any dynamical variables. The Lax representation (\ref{778}) in such a case  can be found in a recent paper \cite{Hone1} by Hone, Roberts, and Vanhaecke.  
In more detail, the matrix relation (\ref{778}) can be written in the form 
\begin{equation}
P_{n+1}=Q_n-P_n,\;\;  Q_{n+1}=\frac{2 P_n - Q_n + R_n}{u_n x}\;\; \mbox{and}\;\; R_{n+1}=u_n x Q_n
\label{79997000988}
\end{equation}
provided we put
\begin{equation}
u_n=p_{1,n}-\frac{q_{1,n}}{2}+\frac{r_{1,n}}{2}.
\label{7999706548}
\end{equation}
Note that the last relation is exactly the same as in Lemma \ref{6753290}. It is also  natural that (\ref{770098}) and (\ref{79997000988}) look exactly the same in form; the difference is only that  in (\ref{770098}) it is supposed that  the coefficients of the polynomials are discrete rational functions of $\left(u_n,\ldots, u_{n+2g-1}\right)$, while in (\ref{79997000988}) they are not.
Under Lemma \ref{6753290}, relation (\ref{7999706548}) is consistent with (\ref{7766548888}), (\ref{7766548}), and (\ref{5643987}) and follows from them.   However, if coefficients of polynomials  $\left(P_n, Q_n, R_n\right)$ are independent of anything, then  (\ref{7999706548}) is merely a  notation for correctly defining the system of discrete equations (\ref{79997000988}).

System (\ref{79997000988}) defines a discrete dynamical system on the phase space $M_g\cong \mathbb{C}^{3g+1}$, whose points are specified by coordinates $\left(q_1,\ldots,q_g; p_1,\ldots,p_g; r_1,\ldots,r_{g+1}\right)$. Probably, it would be correct to call  (\ref{79997000988})  a discrete Mumford dynamical system, since its construction is very similar to the construction of continuous Mumford dynamical systems \cite{Vanhaecke}. By the way, in \cite{Hone1}, the authors call the dynamical system generated by (\ref{79997000988}) a Mumford-like system. 

As was shown in \cite{Hone1},  system (\ref{79997000988}) has remarkable integrability properties. The definition of integrability, in \cite{Hone1}, involves an infinite number of invariants and  compatible Poisson structures on $M_g$.  In conclusion, it is proved in \cite{Hone1} that, for any $g\geq 1$, the Mumford system (\ref{79997000988}) has the property of algebraic integrability.

\subsection{Embedding (\ref{6766438}) into Mumford system (\ref{79997000988})}

Observe, that  the number of parameters $\left(u_0,\ldots, u_{2g-1}; \alpha_0,\ldots, \alpha_g\right)$ that need to be determined in order to fully specify the sequence via (\ref{6766438}) is likewise equal to $3g+1$. All this suggests that the two dynamical systems -- the first of which is the discrete Mumford dynamical system and the second, which is defined by the associated  recurrence  (\ref{6766438}) -- might be equivalent to each other.  At any rate, Theorem \ref{8975465} tells us that any solution of (\ref{6766438}) yields a solution of the system (\ref{79997000988}). This embedding and the corresponding tools are sufficient for us to prove Theorem \ref{76888766}.
In what follows, we will talk about some specific features of this Lax representation, namely its connection with a certain continued fraction.   Additionally, establishing the relationship between (\ref{6766438})  and the Mumford dynamical system furnishes a foundation for the study of their integrability properties.

\subsection{Invariants for discrete Mumford dynamical system}

We do not aim to study the integrability properties of (\ref{6766438}), but for our purposes the following will be needed.
By virtue of (\ref{778}), the matrices $L_n$ are similar to each other for different values of $n\in\mathbb{Z}$. Therefore, 
\[
f=-\det L_n = P_n^2+Q_nR_n
\] 
depends on $x$ but not on $n$. By construction, $f$ is a polynomial in $x$ of the form $f=1+\sum_{j=1}^{2g+1}f_j x^j$, with coefficients
\begin{equation}
f_k=\sum_{j=0}^k p_{j, n}p_{k-j, n}+\sum_{j=0}^{k-1} q_{j, n}r_{k-j, n},
\label{7854387}
\end{equation}
where it is supposed that $p_{0, n}=1$ and  $q_{0, n}=2$. Notice that, in principle, formula (\ref{7854387}) is universal and at first glance does not depend on $g$; however, one must take into account that $p_{j, n}=q_{j, n}=r_{j+1, n}=0$ for all $j\geq g+1$.

Let $y(x)$ be an algebraic function defined by relation 
\begin{equation}
y^2=f(x)=P_n^2+Q_nR_n.
\label{7887657}
\end{equation}
Consider its expansion into a power series $y=1+\sum_{j\geq 1}H_j x^j$. Obviously, the coefficients $H_j$ are polynomials in $\left(f_1,\ldots, f_{2g+1}\right)$ and are essentially invariants for the system (\ref{79997000988}).

It would be interesting, based on the explicit form of the coefficients (\ref{7766548888}), (\ref{7766548}) and  (\ref{5643987}) of the polynomials appearing in the Lax representation, to find the explicit form of the invariants $\left(H_1,\ldots, H_{2g+1}\right)$. However, here we encounter technical difficulties. Nevertheless, one can calculate something. Let us compute, for example, an invariant  $H_1$. For any $g\geq 1$, one has
\begin{eqnarray*}
H_1&=&\frac{f_1}{2}=p_{1, n}+r_{1, n}  \\
&=&T^1_{2g}(n)+\frac{T_{g-1, 1}(n+1)}{\prod_{s=0}^{2g-2} u_{n+s}}+\frac{T_{g, 0}(n)}{\prod_{s=0}^{2g-1} u_{n+s}}.
\end{eqnarray*}
\begin{example}
For example, in the case $g=1$,  $H_1$ has the form
\begin{eqnarray}
H_1&=&u_n+u_{n+1}+\frac{\alpha_1}{u_n}+\frac{\alpha_0+\alpha_1 u_n}{u_nu_{n+1}} \label{675400} \\
&=&\frac{t_nt_{n+3}}{t_{n+1}t_{n+2}}+\frac{t_{n+1}t_{n+4}}{t_{n+2}t_{n+3}}+\alpha_1\left(\frac{t_{n+1}t_{n+2}}{t_nt_{n+3}}+\frac{t_{n+2}t_{n+3}}{t_{n+1}t_{n+4}}\right)+\alpha_0 \frac{t_{n}t_{n+4}}{t_{n+2}t_{n+3}}.
\label{6000075400}
\end{eqnarray}
Expression (\ref{6000075400}) represents a well-known invariant for Somos-5. Relation (\ref{675400}) says that $\left(u_n, u_{n+1}\right)$ lies on the curve defined by the equation
\[
\left(X+Y\right)XY-H_1XY+\alpha_0+\alpha_1\left(X+Y\right)=0.
\]
It turns out that the curve defined by this equation is birationally equivalent to an elliptic curve for which the equation can be written in Weierstrass form. This fact, in turn, is the basis for the analytic form of the terms of the sequence defined by the Somos-5 recurrence. For details, see the paper \cite{Hone2}.

\end{example}

\subsection{Continued fraction}

Here, following \cite{Hone1}, we demonstrate the distinctive feature of the Lax representation (\ref{778}) for the Mumford system, namely, its link to a certain continued fraction.
The following lemma is proved by direct computations. We will see below, it is important for proving the Laurent property for our homogeneous recurrences.
\begin{lemma} \cite{Hone1}
Let $\left(P_n, Q_n, R_n\right)_{n\in\mathbb{Z}}$ be a sequence of triples of polynomials satisfying (\ref{79997000988}), where $u_n$ is given by (\ref{7999706548}). Then the relation 
\begin{equation}
\left(y+P_n\right)\left(y+P_{n+1}\right)=Q_n\left(y+P_{n+1}\right)+u_n x Q_nQ_{n+1},\;\; \forall n\in\mathbb{Z}
\label{66555436}
\end{equation}
holds.
\end{lemma}
\begin{proof}
The proof, as we have already said, proceeds by direct computation. Rewrite (\ref{66555436}) as
\[
y^2+\left(P_n+P_{n+1}-Q_n\right)y+P_nP_{n+1}-Q_nP_{n+1}-u_n x Q_nQ_{n+1}=0.
\]
By the first relation in (\ref{79997000988}) the coefficient of $y$ is zero. And now, taking into account (\ref{7887657}), it remains to be verified that relation (\ref{66555436}) is indeed an identity by virtue of the first and second relations in (\ref{79997000988}).
\end{proof}
Let us denote  
\begin{equation}
F_n=\frac{y+P_n}{Q_n}.
\label{66360097}
\end{equation}
With this notation, we can rewrite (\ref{66555436}) as $F_n=1+u_n x/F_{n+1}$, which in turn yields the remarkable  continued fraction expansion 
\begin{equation}
F_n(x)=St_n(x),
\label{6636}
\end{equation}
where
\begin{equation}
St_n(x)=1+\frac{u_n x}{1+\displaystyle{\frac{u_{n+ 1} x}{1+\displaystyle{\frac{u_{n+ 2} x}{1+\cdots}}}}}.
\label{6639998536}
\end{equation}

To be more precise, we  talk about an infinite sequence of continued fractions $\left(St_n(x)\right)_{n\in\mathbb{Z}}$ defined by (\ref{6639998536}). By definition, this sequence   satisfies a recurrence 
\begin{equation}
St_nSt_{n+1}=St_{n+1}+u_nx.
\label{69998656}
\end{equation}

\subsection{Birational mapping}

Let us expand the continued fraction $St_n(x)$  into a power series
\begin{equation}
St_n(x)=1+\sum_{j\geq 1}(-1)^{j+1} s_{j, n} x^j.
\label{699986000956}
\end{equation}
The coefficients $s_{j, n}$ turn out to be certain homogeneous polynomials in $u_n$.
For example, the first three coefficients $s_{j, n}$ are as follows:
\[
s_{1, n}=u_n,\; s_{2, n}=u_nu_{n+ 1}\;\;\mbox{and}\;\;  s_{3, n}=u_nu_{n+ 1}\left(u_{n+ 1}+u_{n+ 2}\right).
\]
One sees that the coefficient $s_{j, n}$, for any $j\geq 1$, is a polynomial in variables $\left(u_n,\ldots u_{n+j-1}\right)$, but it
is important to note that it depends linearly on the variable $u_{n+j-1}$. This means that $u_{n+j-1}$, for any $j\geq 1$, can be expressed unambiguously as a rational function of the
variables $\left(s_{1, n},\ldots, s_{j, n}\right)$. And this, in turn, implies that  relation (\ref{699986000956}) can be interpreted as a generating one for the birational mapping $\left(u_n, u_{n+1},\ldots\right)\leftrightarrow \left(s_{1, n}, s_{2, n},\ldots\right)$.

Let us  explain why we need this birational map. If the initial data $\left(u_0,\ldots, u_{2g-1}\right)$ for the associated recurrence are given, then, by (\ref{699986000956}), we have a finite set of polynomials $s_j=s_j\left(u_0,\ldots, u_{j-1}\right),\; j=1,\ldots, 2g$.
Denote 
\[
C(x)=St_0(x)-1=\sum_{j\geq 1}(-1)^{j+1} s_{j}x^j.
\] 
\begin{remark}
Here and in what follows, for simplicity, we use the notation $s_j=s_{j, 0}$, $p_j=p_{j, 0}$, $q_j=q_{j, 0}$ and so on.
\end{remark}

Making the substitution $F_0\rightarrow St_0$, we obtain $y=-P+Q\left(1+C\right)$. Substituting the latter into (\ref{7887657}) yields 
\begin{equation}
QC^2+2\left(Q-P\right)C-2P+Q-R=0,
\label{66361111}
\end{equation}
where
\[
P=1+\sum_{j= 1}^g p_{j}x^j,\;\; Q=2+\sum_{j= 1}^g q_{j}x^j,\;\; \mbox{and}\;\;   R=\sum_{j= 1}^{g+1} r_{j}x^j.
\]
\begin{example} \label{78642390}
Let us see how  relation (\ref{66361111}) works in the case $g=1$. Suppose that the coefficients $\left(q_1, p_1, r_1, r_2\right)$ are defined as  (\ref{8754309}), (\ref{875439986509}), (\ref{875476486509777})  and (\ref{875476486509}).  For the initial data one has
\begin{equation}
s_1=u_0=\frac{t_0t_3}{t_1t_2},\; s_2=u_0u_1=\frac{t_0t_4}{t_2^2}
\label{7864299876390}
\end{equation}
and
\begin{equation}
s_m=\left(q_1-p_1\right)s_{m-1}+\sum_{j=1}^{m-1}s_js_{m-j}-\frac{q_1}{2}\sum_{j=1}^{m-2}s_js_{m-j},\;\; m\geq 3.
\label{6644321}
\end{equation}
The coefficients of (\ref{6644321}) belong to the ring $\mathcal{L}_5=\mathbb{Z}[t_0^{\pm 1}, t_1^{\pm 1}, t_2^{\pm 1}, t_3^{\pm 1},  t_{4}^{\pm 1}; \alpha_0, \alpha_1]$, while the initial data $\left(s_1, s_2\right)$ for the recurrence (\ref{6644321}) also belong to this ring. The right-hand side of (\ref{6644321}) contains only multiplications and summations. This means that $s_n\in\mathcal{L}_5$ for all $n\geq 1$.
\end{example}

\subsection{Proof of Theorem \ref{76888766}}

We  present a proof of Theorem \ref{76888766}, essentially following the scheme given in \cite{Hone1}.
Calculations in Example \ref{78642390} can be  generalized for any $g\geq 1$. An finite set of Laurent polynomials 
\begin{equation}
s_j=s_j\left(u_0,\ldots, u_{j-1}\right)=s_j\left(t_0^{\pm 1},\ldots, t_{j+2}^{\pm 1}\right),\;\;j=1,\ldots, g+1
\label{66443887621}
\end{equation}
serves as the initial data for the recurrence relation
\[
s_m=\sum_{j=1}^g(-1)^j\left(p_j-q_j\right) s_{m-j}+\frac{1}{2}\sum_{k=0}^g (-1)^k q_k \sum_{k=1}^{m-k-1} s_j s_{m-j-k},\;\; m\geq g+2
\]
which is derived from (\ref{66361111}). Notice that we put $q_0=2$ here.  An example of the initial data  would be (\ref{7864299876390}). Let us explain once again that the initial data (\ref{66443887621}) are determined by the mapping $\left(u_0, u_{1},\ldots\right)\rightarrow \left(s_{1}, s_{2},\ldots\right)$ defined by a generating relation
\[
St_0(x)=1+\sum_{j\geq 1}(-1)^{j+1} s_{j} x^j,
\]
and further use of substitution  (\ref{676678888768}). 

According to Lemma \ref{6754329876}, the coefficients of this recurrence belong to the ring $\mathcal{L}_{2g+3}$. The initial data also belong to this ring. By induction, we obtain that $s_n\in\mathcal{L}_{2g+3}$ for all $n\geq 1$.

Let us put $\Delta_{-2}=\Delta_{-1}=\Delta_{0}=1$ and define the remaining terms of the sequence $\left(\Delta_n\right)_{n\geq -2}$  as the Hankel determinants: $\Delta_{2k-1}=\det\left(s_{n+m-1}\right)$ and $\Delta_{2k}=\det\left(s_{n+m}\right)$, 
where it is supposed that $n$ and $m$ run from $1$ to $k$.  It is an elementary fact that $\Delta_{n}\in\mathcal{L}_{2g+3},\; \forall n\geq 1$ provided that $s_n\in\mathcal{L}_{2g+3},\; \forall n\geq 1$. Now we need to prove that $t_n\in\mathcal{L}_{2g+3},\;\; \forall n\geq 0$. The following technical lemma will be helpful to this aim.
\begin{lemma}
\label{66654320987}
\cite{Hone1}
The relations
\begin{equation}
t_{2k}=t_0\left(\frac{t_2}{t_0}\right)^k\Delta_{2k-2}\;\; \mbox{and}\;\; t_{2k+1}=\frac{t_0t_1}{t_2}\left(\frac{t_2}{t_0}\right)^{k+1}\Delta_{2k-1},\;\; \forall k\geq 0
\label{986453}
\end{equation}
are identities.
\end{lemma}
\begin{remark}
Note that this lemma has no relation to our recurrences; it simply gives some identities by virtue of  definition of these quantities.
\end{remark}
Relations (\ref{986453})    show that $t_n\in \mathcal{L}_{2g+3},\;\; \forall n\geq 0$, but how to prove the same for $n\leq -1$? Symmetry, which holds for our homogeneous recurrences (\ref{8878886500097}), comes to the rescue -- a property stated in Proposition \ref{004432767}. It tells us that the considerations given in support of the proof of the theorem remain valid if we calculate iterations $t_n$, by virtue of $R_{2g+3}$ in the opposite direction. Thus, we consider the proof of Theorem \ref{76888766} complete.

\section{Integer sequences} \label{7654098756}

Let $\alpha_j=1$ for $j=0,\ldots, g$ and $t_j=1$ for $j=0,\ldots, N-1$. Denote the corresponding sequence of integers generated by $R_{2g+3}$, say,  by the symbol $S_{2g+3}$. Table 1 below shows the first seven terms of the sequence $S_{2g+3}$ following after the ones. For convenience, prime numbers in this table are highlighted in bold.

\begin{table}[H] 
\caption{Some members of the sequence $S_{2g+3}$} 
\begin{tabular}{|c|c|c|c|c|c|c|c|}
\hline
$S_{5}$ & $\mathbf{2}$  & $\mathbf{3}$  & $\mathbf{5}$ & $\mathbf{11}$ & $\mathbf{37}$  & $\mathbf{83}$  & 274   \\
\hline
$S_{7}$  & $\mathbf{5}$  & $\mathbf{13}$   & $\mathbf{61}$ & 185 & 533  & 3973  & 56161   \\
\hline
$S_{9}$ & $\mathbf{13}$ & $\mathbf{73}$  & $\mathbf{937}$  & $\mathbf{5437}$ & $\mathbf{63853}$  & 458353  & $\mathbf{2506657}$     \\
\hline
$S_{11}$ & 34 & $\mathbf{463}$  & 15709  & 215095 & 7193374  & 102225817  & 3088755979    \\
\hline
\end{tabular}
\end{table}
Recall that here $\left(2,\; 5,\; 13,\; 34,\ldots\right)$ that is, the numbers that come first after the ones,  are the Fibonacci numbers with odd indices. This corresponds to Proposition \ref{00098767}.

\section{Remarks and conjectures} \label{786453209}

Let us summarize what was presented in the paper. We exhibited an infinite family of discrete equations possessing the Laurent property and a Lax representation. The idea of the existence of such a family of homogeneous recurrences possessing the Laurent property was proposed at the end of the paper \cite{Svinin6}, and examples were given there. We began our exposition with recurrence (\ref{6766438})
which is fully defined using discrete polynomials $T^k_s(n)$, and it is clear that these discrete polynomials play an important role in the paper.
Now, allow us to spend below some space  to make a number of  remarks and propose some conjectures. We have gathered here material that is not directly related to the main topic of the article.


\subsection{The relationship between the continued fraction $St_{n}(x)$ and the discrete polynomials $T^k_s(n)$}

Let us show how the discrete polynomials $T^k_s(n)$ are directly related to the continued fraction $St_{n}(x)$ given by  (\ref{6639998536}).
To this end, consider the corresponding  convergent 
\begin{equation}
St_{s, n}(x)=1+\frac{u_n x}{1+\displaystyle{\frac{u_{n+ 1} x}{1+\displaystyle{\frac{\cdots}{1+u_{n+s}x}}}}},\; \forall s\geq 0.
\label{675439871}
\end{equation}
It can be shown, that it is related to the discrete polynomials $T^k_s(n)$ by the relation
\begin{equation}
St_{s, n}(x)=\frac{P_{s+1, n}(x)}{P_{s, n+1}(x)},
\label{675439800097}
\end{equation}
where
\[
P_{s, n}(x)=\sum_{j=0}^{\left\lfloor (s+1)/2 \right\rfloor}T^j_{s-j+1}(n)x^j. 
\]
The following lemma is a simple consequence of Lemma \ref{65439876}.
\begin{lemma}
By (\ref{6540987}) and (\ref{67776587}),  the sequence of continuants $P_{s, n}(x)$ satisfies two recurrence relations 
\begin{eqnarray*}
P_{s+2, n}&=&P_{s+1, n}+u_{n+s+1} x P_{s, n} \\
&=&P_{s+1, n+1}+u_{n} x P_{s, n+2}.
\end{eqnarray*}
simultaneously.
\end{lemma}
\begin{proof}
Let us prove, for example, the second relation. From (\ref{675439871}) it obviously follows 
\[
St_{s+1, n}=1+\frac{u_n x}{St_{s, n+1}}.
\]
Substituting (\ref{675439800097}) into this relation yields 
\[
\frac{P_{s+2, n}}{P_{s+1, n+1}}=1+u_n x \frac{P_{s, n+2}}{P_{s+1, n+1}}\;\; \mbox{or}\;\; P_{s+2, n}=P_{s+1, n+1}+u_n x P_{s, n+2}.
\]
\end{proof}

\subsection{Origin of recurrence (\ref{6766438})}

Let us say a few words about where recurrence (\ref{6766438}) comes from. In \cite{Svinin4}, we considered the  equation
\begin{equation}
u_{n+2g+1}=u_n\frac{\sum_{j=0}^g c_j T^{s-j}_{s+j}(n+2)}{\sum_{j=0}^g c_j T^{s-j}_{s+j}(n+1)},
\label{87654999}
\end{equation}
where it is assumed that $\left(c_1,\ldots, c_g\right)$ are  arbitrary  parameters, and $c_0=1$. One can easily  show that (\ref{87654999}) is equivalent to  (\ref{6766438}). The number of parameters $\left(u_0,\ldots, u_{2g-1}; \alpha_0,\ldots, \alpha_g\right)$ needed to uniquely determine the sequence $\left(u_n\right)$ using (\ref{6766438}) is $3g+1$. At the same time, the number of parameters $\left(u_0,\ldots, u_{2g}; c_1,\ldots, c_g\right)$ needed to uniquely determine the sequence defined by (\ref{87654999}) is also $3g+1$, and this is a necessary condition for the equivalence of (\ref{6766438}) and (\ref{87654999}).
It is easy to see that the rational function
\[
K=\frac{\sum_{j=0}^g c_j T^{g-j}_{g+j}(n+1)}{\prod_{j=0}^{2g} u_{n+j}}
\]
is an invariant for (\ref{87654999}). 
The relation 
\[
\prod_{j=0}^{2g} u_{n+j}=\frac{\sum_{j=0}^g c_j T^{g-j}_{g+j}(n+1)}{K},
\]
where $K$ keeping to be an arbitrary constant, actually coincides with (\ref{6766438}). 

\subsection{Homogeneous recurrences of even order}

As already noted, the homogeneous recurrences (\ref{8878886500097}) studied in this article can be represented in the form
\begin{equation}
\prod_{j=3}^{N-3} t_{n+j}\cdot \left(t_nt_{n+N}-\alpha_g t_{n+1}t_{n+N-1}-\alpha_0 t_{n+2}t_{n+N-2}\right)=\sum_{j=1}^{g-1}\alpha_j B^j_{N-j-3}(n+1)
\label{88788}
\end{equation}
with $N=2g+3$. The corresponding discrete polynomials $B^k_s(n)$ were defined in Lemma \ref{7654876} by recurrent relation. From this, in particular, it follows that these recurrences can be reduced to Gale-Robinson recurrences (\ref{887}) of type $\left(N, a, c\right)=\left(2g+3, 1, 2\right),\; \forall g\geq 1$. By the way, one can conjecture that these recurrences represent the maximal possible generalization of the corresponding Gale-Robinson recurrence that preserves the Laurent property.

Note that there is nothing to prevent us from setting $N=2g+2$ in (\ref{88788}). Since all the necessary discrete polynomials $B^k_s(n)$ are defined, we   obtain an infinite family of homogeneous recurrences, which we  denote by the symbol $R_{2g+2}$. In the single case $g=1$, the right-hand side of (\ref{88788}) is equal to zero. Furthermore, we consider the product on the left-hand side of the relation to be equal to one. As a result, we obtain the well-known Somos-4 equation
\[
t_nt_{n+4}=\alpha_1 t_{n+1}t_{n+3}+\alpha_0 t_{n+2}t_{n+2}.
\]
Thus, it is the first representative in the family of even-order recurrences. Let us, for illustration, write down explicitly the recurrence -- the second representative in this family. This is how it is written:
\[
t_nt_{n+6}=\alpha_0 t_{n+2}t_{n+4}+\alpha_1 \frac{t_{n+1}t_{n+4}^2+t_{n+2}^2t_{n+5}}{t_{n+3}}+\alpha_2 t_{n+1}t_{n+5}.
\]
The first question that comes to mind is: do these recurrences have the Laurent property? Experimentation with these recurrences suggests that the answer is most likely positive. Proving this fact is therefore an open problem. 

\subsection{Conjectures about inclusions}

After we have defined the homogeneous recurrences of even order, we can propose some conjectures based on actual computations. Other assumptions are in order.
\begin{conjecture}
Any $R_{2g+2}$-sequence, for every $g\geq 1$, satisfies the recurrence $R_{2g+3}$ with suitable coefficients.
\end{conjecture}
Notice that, in the case $g=1$, this is not a problem. Indeed, it is a well-known fact that any Somos-4 sequence satisfies the Somos-5 recurrence \cite{Hone2}. 
The converse statement is generally not true. By the way, in \cite{Svinin7} there is a much more general result that any Somos-4 sequence satisfies  the Gale-Robinson recurrences (\ref{887}) of every type $\left(N, a, c\right)$ with suitable coefficients. 

The following statement seems to also be true.
\begin{conjecture}
Any $R_{2g+3}$-sequence, for every $g\geq 1$, satisfies the recurrence $R_{2g+4}$ with suitable coefficients.
\end{conjecture}
If these two conjectures are true, then we have the following chain of inclusions:
\[
\mathcal{M}_4\subset \mathcal{M}_5\subset \mathcal{M}_6\subset \cdots,
\]
where $\mathcal{M}_N$ denotes the space of all sequences $\left(t_n\right)$ satisfying the recurrence $R_N$.

\section*{Acknowledgments}

This work was carried out within the framework of  the state assignment of the Ministry of Education
and Science of the Russian Federation on the project No.  126021217175-3.

\end{document}